\documentclass{lmcs}
\pdfoutput=1

\usepackage[utf8]{inputenc}

\usepackage{lastpage}
\lmcsdoi{16}{2}{6}
\lmcsheading{}{\pageref{LastPage}}{}{}%
{Feb.~28,~2019}{May~19,~2020}{}

 \pdfoutput=1
\setlist[description]{leftmargin=10.95003pt}

\usepackage{xspace}
\usepackage{dsfont} 
\usepackage{tikz}
\usetikzlibrary{shapes,shapes.geometric,arrows,fit,calc,positioning,automata}

\usepackage{multicol}

\newcounter{countitems}
\newcounter{nextitemizecount}
\newcommand{\setupcountitems}{%
  \stepcounter{nextitemizecount}%
  \setcounter{countitems}{0}%
  \preto\item{\stepcounter{countitems}}%
}
\makeatletter
\newcommand{\computecountitems}{%
  \edef\@currentlabel{\number\c@countitems}%
  \label{countitems@\number\numexpr\value{nextitemizecount}-1\relax}%
}
\newcommand{\nextitemizecount}{%
  \getrefnumber{countitems@\number\c@nextitemizecount}%
}
\newcommand{\previtemizecount}{%
  \getrefnumber{countitems@\number\numexpr\value{nextitemizecount}-1\relax}%
}
\makeatother    
\newenvironment{AutoMultiColItemize}{%
\ifnumcomp{\nextitemizecount}{>}{3}{\begin{multicols}{2}}{}%
\setupcountitems\begin{itemize}}%
{\end{itemize}%
\unskip\computecountitems\ifnumcomp{\previtemizecount}{>}{3}{\end{multicols}}{}}

\begin{document}

\title[Register Games]{Register Games\rsuper*}

\titlecomment{{\lsuper*}The article extends \cite{lehtinen2018modal} and parts of \cite{BL18}.}
\author[K.~Lehtinen]{Karoliina Lehtinen}	
\address{University of Liverpool, United Kingdom}	
\email{k.lehtinen@liverpool.ac.uk}

\author[U.~Boker]{Udi Boker}	
\address{Interdisciplinary Center (IDC) Herzliya, Israel}	
\email{udiboker@idc.ac.il}  
\thanks{Research supported by the Israel Science Foundation grant 1373/16 and the EPSRC grant EP/P020909/1
(Solving Parity Games in Theory and Practice)}	



\newcommand{\G}{G} 
\newcommand{\V}{V} 
\newcommand{\vi}{v_\iota} 
\newcommand{\pos}{v} 
\newcommand{\posb}{w} 
\newcommand{\E}{E} 
\newcommand{\e}{e} 
\newcommand{\PA}{\Omega} 
\newcommand{\pV}{p} 
\newcommand{\pVV}{q} 
\newcommand{\CD}{I} 
\newcommand{\minp}{i} 
\newcommand{\maxp}{d} 
\newcommand{\play}{\pi} 
\newcommand{\strat}{\sigma}
\newcommand{\stratA}{\tau}
\newcommand{\player}{P} 
\newcommand{\Adam}{Adam\xspace} 
\newcommand{\Eve}{Eve\xspace} 
\newcommand{\rge}[2]{\mathcal{R}^{#1}_{\text{\sc{e}}}(#2)}
\newcommand{\rgo}[2]{\mathcal{R}^{#1}_{\text{\sc{a}}}(#2)}
\newcommand{\rt}{\bar r} 
\newcommand{\WinE}[1]{\text{\sc W}_{\text{\sc{e}}}(#1)}
\newcommand{\WinA}[1]{\text{\sc W}_{\text{\sc{a}}}(#1)}

\newcommand{\tree}{t}
\newcommand{\Gabc}{\Gamma} 
\newcommand{\product}[2]{#1\times #2} 

\newcommand{\class}{C} 

\newcommand{\dc}{dc-size\xspace} 

\newcommand{\pick}{\texttt{choose}}
\newcommand{\update}{\texttt{new}}

\newcommand{\lab}{\mathrm{label}}
\newcommand{\St}{~|~}
\newcommand{\con}{\cdot}
\newcommand{\tuple}[1]{\langle #1  \rangle}
\newcommand{\pair}{\tuple}
\newcommand{\triple}{\tuple}

\newcommand{\A}{{\mathcal A}}
\newcommand{\B}{{\mathcal B}}
\newcommand{\C}{{\mathcal C}}
\newcommand{\D}{{\mathcal D}}
\newcommand{\F}{{\mathcal F}}
\newcommand{\M}{{\mathcal M}}
\renewcommand{\P}{{\mathcal P}}
\newcommand{\R}{{\mathcal R}}
\renewcommand{\S}{{\mathcal S}}
\newcommand{\W}{{\mathcal W}}

\newcommand{\True}{\mathtt{true}}
\newcommand{\False}{\mathtt{false}}

\newcommand{\Nat}{\ensuremath{\mathbb{N}}\xspace}
\newcommand{\Rat}{\ensuremath{\mathbb{Q}}\xspace}
\newcommand{\Reals}{\ensuremath{\mathbb{R}}\xspace}

\renewcommand{\phi}{\varphi}
\newcommand{\limplies}{\to} 
\newcommand{\liff}{\leftrightarrow} 

\newcommand{\Space}{\vspace{0.2cm}}
\newcommand{\Paragraph}[1]{\paragraph{#1.}}

\newcommand{\NotNeeded}[1]{}
\newcommand{\todo}[1]{\textsf{To do: #1}}

\newcommand{\BitZero}{\tt 0} \newcommand{\BitOne}{\tt 1} \newcommand{\BitSep}{\tt \$}

\newcommand{\win}[1]{\mathsf{Win}_{#1}}
\newcommand{\Nwin}[1]{\overline{\win{#1}}}
\newcommand{\game}[2]{\ensuremath{\mathcal{G}(#1 , #2)}\xspace} 

\newcommand{\AutTuple}{\tuple}
\newcommand{\MaybeNeeded}[1]{}

\newcommand{\Func}[1]{{\mathsf{#1}}}
\newcommand{\move}{\Func{move}}
\newcommand{\reset}{\Func{reset}}
\newcommand{\Suf}{\Func{suffixes}}
\newcommand{\Subtrees}{\Func{subtrees}}
\newcommand{\Bit}{\Func{bit}}
\newcommand{\Next}{\Func{Next}}
\newcommand{\NextIndex}{\Func{NextIndex}}
\newcommand{\BP}{\Func{B}^+}



\begin{abstract}
The complexity of parity games is a long standing open problem that saw a major breakthrough in 2017 when two quasi-polynomial algorithms were published. 

This article presents a third, independent approach to solving parity games in quasi-polynomial time, based on the notion of \emph{register game}, a parameterised variant of a parity game.
The analysis of register games leads to a quasi-polynomial algorithm for parity games, a polynomial algorithm for restricted classes of parity games and a novel measure of complexity, the \textit{register index}, which aims to capture the combined complexity of the priority assignement and the underlying game graph.

 We further present a translation of alternating parity word automata into alternating weak automata with only a quasi-polynomial increase in size, based on register games---this improves on the previous exponential translation. 

We also use register games to investigate the parity index hierarchy: while for words the index hierarchy of alternating parity automata collapses to the weak level, and for trees it is strict, for structures between trees and words, it collapses \emph{logarithmically}, in the sense that any parity tree automaton of size $n$ is equivalent, on these particular classes of structures, to an automaton with a number of priorities logarithmic in $n$.
\end{abstract}

\maketitle

\section{Introduction}\label{sec:Introduction}

A play in a parity game consists of a player, whom we shall call Eve, and her opponent, Adam, moving a token along the edges of a graph labelled with integer priorities, forever, thus forming an infinite path. Eve's objective is to force the highest priority that occurs infinitely often to be even, while Adam tries to stop her.

These games arise at the intersection of logic, games, and automata theory. In particular, they are the acceptance games for alternating parity automata, both on trees and on $\omega$-words.
The complexity of solving parity games---that is, of deciding which player has a winning strategy---still is, despite extended efforts, an open problem: it is in {\sc UP}$ \cap ${\sc coUP}~\cite{jurdzinski1998deciding} yet it is not known to admit a polynomial algorithm. After over twenty-five years of incremental improvements, Calude, Jain, Khoussainov, Li, and Stephan published the first quasi-polynomial solution~\cite{calude2017deciding}. Only a little later in the same year, Jurdzi\'nski and Lazi\'c presented an independent progress-measure based algorithm that achieves the same complexity~\cite{jurdzinski2017succinct}.
This article presents a third, independent approach to solving parity games in quasi-polynomial time. The automata-theoretic method also solves the more general problem of translating alternating parity word automata into alternating weak automata with quasi-polynomial increase in state space, and offers some new insights into the descriptive complexity of parity games---that is, the complexity of formalisms that can recognise winning regions in parity games---and the parity index problem.

\subsubsection*{Register Games} Our first contribution is to present \emph{register games}, a parameterised variant of parity games, that we then use to analyse the complexity of parity games and parity automata, both on infinite words and infinite trees.
A register game consists of a normal parity game, augmented with a fixed number of registers that keep partial record of the history of the game. Although the register game is harder for Eve than the parity game on the same arena, if she can win the parity game, then she can also win the register game as long as she has a large enough number of registers. Exactly how many registers she needs depends on the parity game arena. We call this the register-index of the parity game: it is a measure of complexity which takes into account both the priority assignment of the parity game and the structure of the underlying graph.

Two key properties of register games then enable us to derive complexity results for finite  parity games as well as for parity automata on both words and restricted classes of trees. First, register games are \emph{automata-definable}: there is, for every integer $k$, an alternating parity tree automaton that accepts infinite parity games of fixed maximal priority in which Eve wins the $k$ register game; furthermore, for any alternating parity automaton $\A$, we can define a family of parameterised automata $\A_k$ that use the $k$-register game as their acceptance game instead of standard parity games.
Second, the register-index is bounded \emph{logarithmically} in the number of disjoint cycles of a parity game arena. 

\subsubsection*{A quasi-polynomial algorithm for parity games} For finite arenas, where the number of disjoint cycles is bounded by the number of positions, solving parity games reduces to solving the register game with a number of registers logarithmic in the size of the game. This results in a quasi-polynomial parity game algorithm with running time in $2^{O((\log n)^3)}$, as well as a parameterised polynomial algorithm that solves classes of parity games with bounded register-index.

\subsubsection*{Word automata transformations} The complexity of solving parity games is intimately related to the complexity of turning alternating parity word automata (APW) into alternating weak word automata (AWW). Indeed, solving parity games amounts to checking the emptiness of an APW on the trivial singleton alphabet. Since the emptiness of AWW on the singleton alphabet can be checked in linear time~\cite{KV98b}, a translation from APW to AWW immediately yields an algorithm for solving parity games of which the time complexity matches the size-increase and time-complexity of the automata translation.
 Note, however, that the automata-translation question is more general, since automata need not be defined over a one-letter alphabet, and even a binary alphabet can add substantial complexity.
 
Nevertheless, until 2017, the best known algorithms for the two problems were roughly the same: exponential in the number of priorities. A competitive tool for solving parity games is even based on the translation from parity to weak automata \cite{di2016solving}.
In 2017, however, the advent of quasi-polynomial algorithms created a gap between the complexity of solving parity games, and the automata translation.

In this article, we show that the analysis of parity games that are infinite, or at least of unbounded size, allows us to generalise the quasi-polynomial time complexity of solving parity games to the blow-up incurred when turning alternating parity word automata into alternating weak automata\footnote{Recently a translation based on universal trees and infinite progress measures has further improved this upper bound to $n^{O(\log \frac{d}{\log n})}$\cite{DJL19}.}.

\subsubsection*{The index hierarchy}While the translation of alternating parity automata into weak is always possible for word automata~\cite{KV01c}, this is not the case for tree automata. Indeed, while alternating weak automata suffice to capture all $\omega$-regular word languages, no fixed number of priorities suffices to capture all regular tree languages---this follows from the equivalence between alternating parity automata and the modal $\mu$-calculus~\cite{wilke2001alternating}, and the strictness of the modal $\mu$-calculus alternation hierarchy~\cite{bradfield98strict}. We say that the \emph{parity index hierarchy} is strict on trees, but collapses to the weak level over words. 

We study automata on structures that are, in some sense, \emph{between} words and trees and show, using register games, a \emph{logarithmic} collapse of the index hierarchy: for every alternating parity automaton $\A$ with $n$ states, there is an alternating parity automaton $\A'$ with only $O(\log n)$ priorities that is equivalent to $\A$ over these structures.

\subsubsection*{Of other quasi-polynomial automata} \textit{Separating automata} have been proposed as a way to understand the underlying combinatorial structure of the different quasi-polynomial algorithms~\cite{toolbox}. We conclude with a discussion on what it would take for other quasi-polynomial algorithms for parity games, when seen as {separating automata}, to be extended into translations of alternating parity automata into weak automata.\\

This article is based on Lehtinen's quasi-polynomial algorithm for parity games~\cite{lehtinen2018modal}, which introduced the notion of register games, and Boker and Lehtinen's extension of this technique into a translation of alternating parity word automata into weak automata~\cite{BL18}. Here the definition of register games is simplified, generalised to infinite parity games, and presented from an automata-theoretic, rather than modal $\mu$-calculus, perspective.




\section{Parity Games}\label{sec:Parity-Games}
 
\begin{defi}[Parity games]
A parity game is an infinite-duration two-player zero-sum path-forming game, played between Eve and her opponent Adam on a potentially infinite game graph $\G =(\V,\V_{\text{\sc e}},\V_{\text{\sc a}},\E,\PA)$ called the \textit{arena}. The \emph{positions} $\V$ of the arena are partitioned into those belonging to \Eve, $\V_{\text{\sc e}}$, and those belonging to \Adam, $\V_{\text{\sc a}}$. The \textit{priority assignment} $\PA: \V \rightarrow \CD$ maps every position in $V$
to a \textit{priority} in a finite co-domain $\CD=\{\minp,\minp+1..,\maxp\}$ where $\minp\in \{0,1\}$. The edge-relation $E\subseteq \V \times \V$ defines the successors of each position. Without loss of generality we assume all positions to have at least one successor. On finite parity games of size $n$, we can assume $\maxp\leq n$.

A play is an infinite sequence of positions $\play=\pos_0\pos_1...$ such that $(\pos_{i},\pos_{i+1})\in E$ for all~$i\geq 0$. 
%
%
A play $\play$ is winning for Eve if the highest priority that occurs infinitely often along $\play$ is even; otherwise it is winning for Adam.

A (positional) \textit{strategy} $\strat$ for a player $\player\in \{\text{\Adam,\Eve}\}$ in a parity game $\G$ maps every position $\pos$ belonging to $\player$ in $\G$  to one of its successors.
 A play is said to agree with a strategy $\strat$ for $\player$ if $\pos_{i+1}=\strat(\pos_i)$ whenever $\pos_i$ belongs to $\player$.
A strategy $\strat$ for player $\player$ is said to be winning  for $\player$ from a position $\pos$ if all plays starting at $\pos$ that agree with $\strat$ are winning for~$\player$. We call the positions from which \Eve has a winning strategy \Eve's winning region in $\G$, written $\WinE{\G}$; \Adam's winning region is written $\WinA{\G}$.

 We write $\G,\pos$ for the parity game $\G$ with a designated initial position $\pos$. A winning strategy in $\G,\pos$ is a strategy that is winning from $\pos$.
\end{defi}

\begin{thm}[Positional Determinacy~\cite{emerson1991tree,mostowski1991games}]
In all positions in a parity game, one of the players has a positional winning strategy.
\end{thm}

It will sometimes be convenient, for clarity and aesthetics,  to assign priorities to edges, with $\Omega:E\rightarrow I$. A parity game with edge priorities can be converted into one with vertex priorities by introducing intermediate, priority-carrying nodes onto edges and giving a low priority to other vertices. Conversely, a vertex-labelled parity game can be converted into an edge-labelled one by assigning the priority of a vertex to its outgoing edges.

\section{Register Games}\label{sec:RegisterGames}

This section describes the key technical development of this article: register games. These are parameterised variations of parity games, also played on a parity game arena. 
 Crucially, the winning condition of a $k$-register game is a parity condition that ranges over priorities 
 $[0..2k+1]$
 rather than the priorities of the arena.
 The larger the parameter $k$, the easier the $k$-register game becomes for Eve who, on arenas in which she has a winning \textit{parity game} strategy, is guaranteed to also have a winning  $k$-register game strategy, for some large enough $k$. Note that the mechanics of register games have been simplified, compared to their first appearance in~\cite{lehtinen2018modal}; see Remark \ref{rem:diff}.
 
\subsection{Definitions and Observations}
Informally, the $k$-\textit{register game} 
 consists of a normal parity game, augmented with a 
 tuple $(r_0,...,r_{k})$
 of registers that keeps a partial record of the history of the game. During a turn, several things happen: the player whose turn it is in the parity game moves onto a successor position of their choice, which has some priority $p$ and Eve chooses an index $i$, 
 $0\leq i\leq {k}$; 
 then, the registers get updated according to both~$i$ and $p$, and an output between $0$ and $2k+1$ is produced, also according to $i$ and $p$.
 
 The update wipes out the contents of registers with index lower than $i$: for $j< i$, $r_j$ is set to $0$. Meanwhile $r_i$ is set to $p$ and $r_j$ for $j>i$ to $\max(r_j,p)$. In other words, each register $r_j$ keeps track of the highest priority seen since Eve last chose $i$ with~$i\geq j$.
 The output is $2i$ if $\max(r_i,p)$ is even, and $2i+1$ otherwise.
%
 Then, in the limit, Eve wins a play if the largest output that occurs infinitely often is even.

 Since the winning condition of the register game is a parity condition, we can formally define the $k$-register game on a parity game arena $\G$ as a parity game on an arena $\rge{k}{\G}$, of which the positions are positions of $\G$ paired with vectors in 
 $\CD^{k+1}$
  that represent the contents of the registers. An additional binary variable $t$ indicates whether the next move consists of Eve's choice of register ($t=0$), or a move in the underlying parity game ($t=1$).
 
\begin{defi}[Register game]\label{D:register-game}
Let $\G$ be a parity game  $(\V^\G,\V^\G_{\text{\sc e}},\V^\G_{\text{\sc a}},\E^\G,\PA^\G)$ and let $\CD$ be the co-domain of $\PA^\G:\V^\G\rightarrow \CD$. For a fixed parameter $k\in \mathds{N}$, the arena of the $k$-register game $\rge{k}{\G}$ on $\G$ in which \Eve controls the registers, consists of $\rge{k}{\G}=(\V,\V_{\text{\sc e}},\V_{\text{\sc a}},\E,\PA)$ as follows. 

While $\G$ carries its priorities on its vertices, for the sake of clarity, $\rge{k}{\G}$ carries them on its edges, $\PA : \E\rightarrow [0..2k+1]$.

\begin{itemize}
\item $\V$ is a set of positions $(\pos,\rt,t)\in \V^\G \times \CD^{k+1}\times \{0,1\}$,
\item $\V_{\text{\sc a}}$ consists of $(\pos,\rt,1)$ such that $\pos\in V^G_{\text{\sc a}}$,
\item $\V_{\text{\sc e}}$ consists of $\V\setminus \V_{\text{\sc a}}$,
\item $\E$ is the disjoint union of sets of edges $\E_{\mathit{move}}$ and $E_i$ for all $i\in [0..k]$ where:

$\E_{\mathit{move}}$ consists of edges $((\pos,\rt, 1),(\posb,\rt,0))$ such that $(\pos,\posb)\in \E^\G$.

For each $i\in [0.. k]$, $\E_i$ consists of edges $((\pos,\rt, 0),(\pos,\rt',1))$ such that:
\begin{itemize}
\item $r'_j=\max(r_j,\PA^\G(\pos))$ for $j>i$,
\item $r'_j=\PA^\G(\pos)$ for $j=i$, and
\item $r'_j=0$ for $j< i$.
\end{itemize}
\item $\Omega$ assigns priorities from $[0..2k+1]$ to edges as follows:
\begin{itemize}
\item Edges of $E_{\mathit{move}}$ have priority $0$;
\item $((p,\rt, 0),(p,\rt',1))\in E_i$ has priority $2i$ if $\max(r_i,p)$ is even, and priority $2i+1$ otherwise.
\end{itemize}
\end{itemize}

\subsubsection*{Terminology} Given a play in $\rge{k}{\G}$, we call the \textit{underlying play} its projection onto the first element of each visited position. At a position $(\pos,\bar r,t)$, we write that: a priority $\pV\in \CD$ \textit{occurs}  if $\PA^\G(\pos)=\pV$; a register $i\in [0.. k]$ \emph{contains} a priority $\pV\in \CD$ if $r_i=\pV$; Eve \emph{chooses} register $i$ and \emph{outputs} $j\in [0.. 2k+1]$ if the play follows an edge in $\E_i$ of priority $j$. \\

 A strategy for \Adam in $\G$ induces a strategy for \Adam in $\rge{k}{G}$. A strategy for \Eve in $\G$ paired with a register-choosing strategy in $\rge{k}{G}$ induces a strategy for \Eve in $\rge{k}{G}$.
 Observe that as the winning condition depends on the outputs that occur infinitely often, and that registers up to the largest one chosen infinitely often renew their contents infinitely often, if a player has a winning strategy in $\rge{k}{\G}$ from $(\pos, \rt, t)$, then they have a winning strategy from all $(\pos, \mathunderscore\,,\mathunderscore)$. We will then simply say that they have a winning strategy from $\pos$.
\end{defi}

The $k$-register game arena $\rgo{k}{\G}$ where \Adam controls the registers is similar to $\rge{k}{\G}$ except that positions $(\pos,\rt,0)$ are in $\V_{\text{\sc a}}$, edges in $\E_{\mathit{move}}$ have priority $0$, and edges $((\pos,\rt, 0),(\pos,\rt',1))$ of $\E_i$ have priority $2i+2$ if $r_i$ is even, and  $2i+1$ otherwise. The $k$-register game with \Eve (resp., \Adam) in control of registers on an arena $\G$ is the parity game on the arena $\rge{k}{\G}$ (resp., $\rgo{k}{\G}$).
Unless specified, the $k$-register game on $\G$ refers to  $\rge{k}{\G}$.

For fixed $k$, $\rge{k}{\G}$ is a parity game on an arena of size polynomial in the size of $\G$ and of priority domain $[0..2k+1]$: it can be solved in polynomial time in the size of $\G$ using any solver exponential in the number of priorities (e.g. a progress measure algorithm~\cite{jurdzinski2000small}).\\

 We now establish two important facts about register games: if \Adam has a winning strategy from a position $\pos$ in a parity game $\G$, then he has a winning strategy from all positions $(\pos,\rt,t)$ in $\rge{k}{\G}$ for any $k$. However, if \Eve has a winning strategy in $\G$ from $\pos$, then she also has a winning strategy from all $(\pos,\rt,t)$ in $\rge{k}{\G}$ as long as $k$ is large enough.

 \begin{lem}
 If \Adam has a winning strategy in the parity game $\G$ starting at $\pos$, then he also has a winning strategy in $\rge{k}{\G}$ starting at $\pos$
 for all $k$.
 \end{lem}
 
 \begin{proof}
Assume \Adam has a strategy $\stratA$ in $\G$ that is winning from $\pos$. On any play $\play$ starting at $\pos$
that agrees with $\stratA$, the highest priority $\pV$ seen infinitely often is odd. Let $i$ be the highest register chosen infinitely often by \Eve during $\play$. Then, eventually---that is, after the last occurrence of anything higher than $\pV$ and after \Eve no longer chooses registers higher than $i$---whenever $\pV$ is seen, $r_i$ is set to $\pV$ and remains at $\pV$ until Eve again chooses~$i$. Then, since $\pV$ is odd, this outputs $2i+1$. Since $\pV$ occurs infinitely often and $i$ is picked infinitely often, $2i+1$ is output infinitely often. It is also the highest value output infinitely often because $i$ is the highest index picked infinitely often. Since $2i+1$ is odd, $\play$ is winning for \Adam. The strategy $\stratA$ is therefore winning for \Adam in $\rge{k}{\G}$ from  position $\pos$.
 \end{proof}

 \begin{lem}
 If \Eve has a winning strategy in the parity game $\G$ with parity co-domain $\CD$ at a position $\pos_\iota$, then she also has a winning strategy in $\rge{k}{\G}$ from 
$\pos_\iota$ 
 for $k\geq i$ where $2i$ is the largest even priority in $\CD$.
 \end{lem}
 
 \begin{proof}
 Given a strategy $\strat$ in the parity game $\G$ that is winning for \Eve from $\pos_\iota$, let $\strat'$ be the following strategy for \Eve in $\rge{k}{\G}$ where $k\geq i$ and $2i$ is the largest even priority in $\CD$: at positions $(\pos,\rt,1)$ the strategy $\strat'$ follows $\strat$, that is if $\pos$ belongs to \Eve, then $\strat'(\pos,\rt,1)=(\posb,\rt,0)$ where $\strat(\pos)=\posb$; at positions $(\pos,\rt,0)$ where $\PA^\G(\pos)=2i$ or $2i+1$ for some $i$, the strategy $\strat'$ chooses register $i$.
 
 We now argue that $\strat'$ is winning for \Eve in $\rge{k}{\G}$ from $\pos_\iota$. Since $\strat'$ follows $\strat$ in the underlying game, the highest priority $\pV$ that occurs on any play beginning at $\pos_{\iota}$ that agrees with $\strat'$ is even, say $\pV=2i$ for some $i$. The highest register chosen infinitely often is therefore $i$. Since eventually nothing higher than $\pV$ occurs anymore, $r_{i}$ will eventually remain $\pV$ in perpetuity, and therefore, eventually, every time \Eve chooses $i$, this outputs $2i$; since this is the highest register chosen infinitely often, nothing higher is output infinitely often. Hence \Eve wins every play that agrees with $\strat'$ from $\pos_\iota$. 
 \end{proof}
 

The number of registers that \Eve needs to win the register game on a parity game in which she has a winning strategy depends on the complexity of her winning strategy. We define this as the \emph{register index} of a parity game and consider it as a measure of complexity for parity games, comparable to measures such as entanglement~\cite{berwanger2012entanglement}.

 \begin{defi}[Register-index]
A parity game $\G$  has register-index $k$ at $\pos\in \WinE{\G}$ if $k$ is the smallest integer such that \Eve wins $\rge{k}{\G}$ from $\pos$,
and at $\pos\in \WinA{\G}$ if $k$ is the smallest integer such that \Adam wins $\rgo{k}{\G}$ from 
$\pos$. A parity game $\G$ has register index $k$ if $k$ is the minimal integer such that it has register index up to $k$ at all positions.
 \end{defi}

\begin{rem}\label{rem:diff}
The register game defined here differs slightly from the one defined in~\cite{lehtinen2018modal} and used in~\cite{BL18}. Here, for a more elegant presentation, Eve chooses a register at every turn, but has an additional $0$-indexed register she can default to. This avoids having  an additional priority that encodes that Eve must reset infinitely often. Furthermore, and perhaps  more significantly, the register update mechanism is simplified: instead of values shifting between registers, all registers below the chosen one get reset to $0$ while the chosen register updates to the current priority.
This new game mechanism simplifies Eve's strategies and in particular the proof of Theorem~\ref{lem-log}.
 Finally, we do not restrict register games to finite arenas, but consider them on potentially infinite ones.
\end{rem}

\subsection{Examples}\label{sec:examples} In this section we explore which features of parity games affect the register index, and which do not.
Since the register-index depends on a player's winning regions, the examples presented in this section are all on one-player games: all positions belong to \Adam. They can  be embedded into two-player games of arbitrary complexity; however, for the register-index, only the game induced by the simplest winning strategy matters.

We begin by looking at register games with a small number of registers, starting with $0$-register games.
We observe that \Eve wins the $0$-register game if she has a strategy $\sigma$ such that plays that agree with $\sigma$ see finitely many odd priorities, such as the game depicted in Figure \ref{fig:0}. Indeed, if \Eve follows such a strategy, then eventually all outputs in the $0$-register games are $0$.
Among well-known parity games with register-index $0$ are the example games for which strategy improvement and divide-and-conquer algorithms exhibit worst-case complexity~\cite{benerecetti2017robust,friedmann2009exponential}.

While parity games with register-index $0$ must be quite simple, register-index $1$ already captures games with more complexity. Example \ref{ex-small} for instance dicusses games with high entanglement, high tree-width, and a large number of priorities that have register-index $1$. Known examples of families of parity games of register-index $1$ are those that exhibit worst-case complexity for Zielonka's recursive algorithm and the quasi-polynomial progress measure algorithms~\cite{friedmann2011recursive,fearnley2017ordered}. In these games every odd priority is immediately followed by a larger even priority, so \Eve still has an easy $1$-register game strategy consisting of choosing register $1$ whenever an even priority occurs, and register $0$ otherwise. 

\begin{exa}\label{ex-small}
Figure \ref{fig-high} shows an edge-labelled arena in which \Eve wins the parity game but loses the $0$-register game.
In the $0$-register game, \Adam's strategy is to loop at the current position once (this clears the register), then move to the other position and repeat; \Eve has no choice but to produce outputs $1$ and $0$ infinitely often.
\Eve can win the $1$-register game by choosing register $1$ after seeing $2$, and register $0$ after seeing $1$ or $0$.

Figure \ref{fig-1} illustrates a slightly more complicated family of parity games, which has linear entanglement, tree-width and number of priorities, yet constant register-index $1$. It consists of arenas with vertices $v_0...v_n$ with priority edges $(v_j,v_i)$ of priority $2i$ for $j\leq i$ and $2j-1$ for~$j>i$.
\Eve's strategy is to choose register $1$ whenever she sees an even priority and register $0$ otherwise. Since odd priorities only occur after a larger even priority, with this strategy register $1$ permanently contains even priorities. Then every occurrence of an even priority leads to output $2$, while odd priorities lead to output $1$.
\end{exa}
\begin{figure}
\begin{minipage}[t]{.5\linewidth}\centering
{
\vspace{-6.5cm}
\begin{tikzpicture}
[->,>=stealth',shorten >=1pt,auto,node distance=2cm,
                    semithick]
  \tikzstyle{every state}=[fill=none,draw=black,text=black]
  
  \node[state,initial] (A) { };
  \node[state] (B) [right of=A] { };
  \node[state] (C) [right of=B] { };

  \path (A) edge [loop above] node {$0$} (A)
        (B) edge [loop above] node {$2$} (B)
        (C) edge [loop above] node {$0$} (C)

        (A) edge [bend left] node {$1$} (B)
        (B) edge [bend left] node {$3$} (C);        

\end{tikzpicture}\caption{Parity game of register-index $0$}\label{fig:0}

\vspace{1.005cm}
\begin{tikzpicture}
[->,>=stealth',shorten >=1pt,auto,node distance=2cm,
                    semithick]
  \tikzstyle{every state}=[fill=none,draw=black,text=black]
  
  \node[state,initial] (A) { };
  \node[state] (B) [right of=A] { };
    
  \path (A) edge [loop above] node {$0$} (A)
        (B) edge [loop above] node {$0$} (B)
        (A) edge [bend left] node {$1$} (B)
        (B) edge [bend left] node {$2$} (A);        

\end{tikzpicture}\caption{Parity game of register-index $1$}\label{fig-high}
}

\end{minipage}%
\hfill
\begin{minipage}[t]{.45\textwidth}\centering

  {
\begin{tikzpicture}[->,>=stealth',shorten >=1pt,auto,node distance=4cm,
                    semithick]
  \tikzstyle{every state}=[fill=none,draw=black,text=black]

  \node[state] (A)                    {$v_3$};
  \node[state]         (B) [right of = A] {$v_2$};
  \node[state]         (C) [below of = A] {$v_0$};
  \node[state]         (D) [right of = C] {$v_1$};;

  \path (A) edge [loop above] node {$6$} (A)
        (C) edge [loop below] node {$0$} (C)
        (B) edge [loop above] node {$4$} (B)
        (D) edge [loop below] node {$2$} (D)
        (A) edge [bend right=10] node [below] {$5$} (B)
        (B) edge [bend right=10] node [above] {$6$} (A)
        (B) edge [bend right=10] node [left] {$3$} (D)
        (D) edge [bend right=10] node [right] {$4$} (B)
        (C) edge [bend right=10] node [below]{$2$} (D)
        (D) edge [bend right=10] node [above]{$1$} (C)
        (A) edge [bend right=10] node [left]{$5$} (C)
        (C) edge [bend right=10] node [right]{$6$} (A)
        (A) edge [bend right=17] node [left]{$5$} (D)
        (D) edge [bend right=17] node [right]{$6$} (A)
        (B) edge [bend right=17] node [left]{$3$} (C)
        (C) edge [bend right=17] node [right]{$4$} (B);
\end{tikzpicture}\caption{Parity game of register-index $1$}\label{fig-1}

}
\end{minipage}
\end{figure}

We have seen some parity games of low register-index. Building parity games of high register-index is more involved and requires, as we shall see, exponentially many positions.

\begin{lem}\label{example-high}
For all $n$, there exists a parity game of register-index at least $n$.
\end{lem}

\begin{proof}
Let $H_0$ be the game arena consisting of a single node, belonging to \Adam, with a self-loop of priority $0$. This unique node is also the initial node of $H_0$.

Then, for all $n>0$, the arena $H_{n}$ consists of two distinct copies of $H_{n-1}$ with initial positions $v_0$ and $v_1$
respectively, with an edge $(v_0,v_1)$ of priority $2n-1$ and an edge $(v_1,v_0)$ of priority $2n$. The position $v_0$ is also the initial position of $H_n$. See Figure \ref{fig-family}.

 \Eve wins these parity games $H_n$ because all cycles are dominated by an even priority. We will show that for $n>0$, \Adam has a winning strategy in the $n-1$-register game~$\rge{n-1}{H_n}$. 

We reason inductively, and show that for each $H_n$, $n>0$, that i) from the initial position in $\rge{m}{H_n}$ for $m\geq n-1$, with a register configuration in which register contents are bounded by $2n-1$, \Adam can force the game to output $2n-1$ or a higher odd priority, before returning to the initial position, and ii) \Eve loses in $\rge{n-1}{H_n}$.

\begin{description}
\item[Base case] \Adam has a winning strategy in $\rge{0}{H_1}$: His strategy is to loop once in the current position to see $0$---this sets the register content to $0$; he then moves to the other position and repeats. This causes both $0$ and $1$ to be output infinitely often.

If \Adam uses this strategy in $\rge{m}{H_1}$ for $m\geq 0$, starting from a register configuration in which register contents are bounded by $1$, although he can't win, he can force the game to output $1$ or a higher odd priority before returning to the initial position.

\item[Inductive step] Assume i) and ii) for $H_n$.

i) Consider the following strategy for \Adam in $\rge{m}{H_{n+1}}$ for $m\geq n$. He first moves from $v$, the initial position of $H_{n+1}$ onto the initial position $v'$ of the second component of $H_{n+1}$, via the odd priority $2n+1$.
Then, he plays in $H_n$, which only contains priorities smaller than $2n+1$, with a strategy that is winning in $\rge{n-1}{H_n}$. To counter this strategy, \Eve has to eventually choose a register of index $n$ or higher, after which \Adam returns to the initial position. This is his strategy $\tau_{n}$. Observe that if at the initial position all register contents are bounded by $2n+1$, then \Eve will either lose in the second $H_n$ component, or choose a register of index $n$ or higher when it contains the odd priority $2n+1$, outputting $2n+1$ or a higher odd priority.

ii) We now show that he also has a winning strategy in $\rge{n}{H_{n+1}}$. He begins by playing $\tau_{n}$ until $2n+1$ is output and the play is back at the initial position. Note that some registers now might contain $2n+2$,
so he can not yet repeat $\tau_{n}$. Instead, he  plays a strategy that is winning in $\rge{n-1}{H_n}$ in
the \textit{first} $H_n$ component of $H_{n+1}$. Again, \Eve will lose unless she chooses register $n$. After she has chosen the register $n$, all the registers hold values smaller than $2n+1$. \Adam can then return to the initial position, and again use $\tau_{n}$ to force output $2n+1$. Thus by alternating between using $\tau_{n}$ to force the maximal odd output, and clearing the registers of higher priorities in order to be able to use $\tau_{n+1}$ again, \Adam forces $2n+1$ to be output infinitely often. 

Hence $H_n$ has register-index at least $n$.
\qedhere
\end{description}
\end{proof}

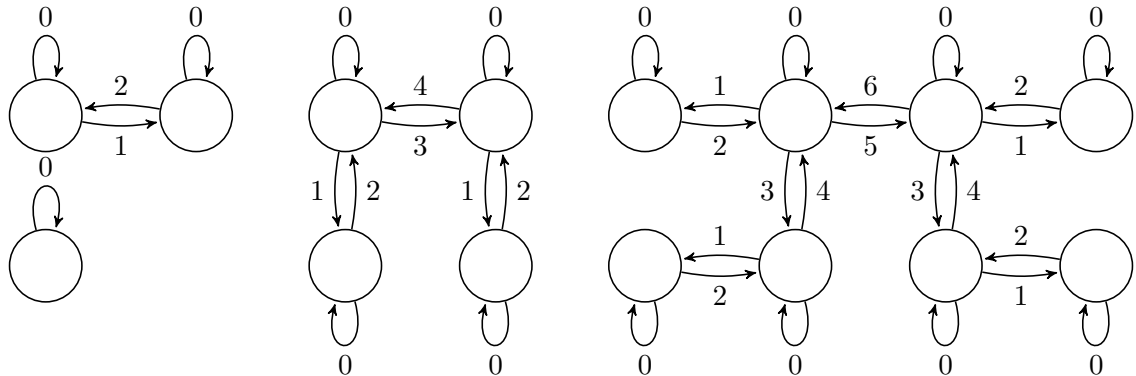
\begin{figure}
\begin{center}

\begin{tikzpicture}[->,>=stealth',shorten >=1pt,auto,node distance=2cm,
                    semithick]
  \tikzstyle{every state}=[fill=none,draw=black,text=black]
  \node[state] (A1)         {};
  \node[state]         (B1) [right of = A1] {};
    \path (A1) edge [loop above] node {$0$} (A1)
        (B1) edge [loop above] node {$0$} (B1)
        (A1) edge [bend right=10] node [below] {$1$} (B1)
        (B1) edge [bend right=10] node [above] {$2$} (A1);

 \node[state] (F) [below of =A1] {};
 \path (F) edge [loop above] node {$0$} (F);
 
  \node[state] (A2)   [right=1cm of B1]                 {};
  \node[state]         (B2) [right of = A2] {};
  \node[state]         (C2) [below of = A2] {};
  \node[state]         (D2) [right of = C2] {};
  
      \path (A2) edge [loop above] node {$0$} (A2)
        (B2) edge [loop above] node {$0$} (B2)
        (C2) edge [loop below] node {$0$} (C2)
        (D2) edge [loop below] node {$0$} (D2)
        (A2) edge [bend right=10] node [below] {$3$} (B2)
        (B2) edge [bend right=10] node [above] {$4$} (A2)
        (A2) edge [bend right=10] node [left] {$1$} (C2)
        (C2) edge [bend right=10] node [right] {$2$} (A2)
        (B2) edge [bend right=10] node [left] {$1$} (D2)
        (D2) edge [bend right=10] node [right] {$2$} (B2);

  \node[state]         (B) [right=7cm of A2] {};
  \node[state] (A)   [left of = B]    {};
  \node[state] (E)         [left of =A] {};
  \node[state]         (C) [below of = A] {};
  \node[state]         (D) [right of = C] {};

  \node[state] (F)         [left of = C] {};
  \node[state] (G)         [right of = B] {};
  \node[state] (H)         [right of = D] {};

  \path (A) edge [loop above] node {$0$} (A)
        (C) edge [loop below] node {$0$} (C)
        (B) edge [loop above] node {$0$} (B)
        (D) edge [loop below] node {$0$} (D)
        (E) edge [loop above] node {$0$} (E)
        (F) edge [loop below] node {$0$} (F)
        (G) edge [loop above] node {$0$} (G)
        (H) edge [loop below] node {$0$} (H)
        (A) edge [bend right=10] node [below] {$5$} (B)
        (B) edge [bend right=10] node [above] {$6$} (A)
        (B) edge [bend right=10] node [left] {$3$} (D)
        (D) edge [bend right=10] node [right] {$4$} (B)
  
        (A) edge [bend right=10] node [left]{$3$} (C)
        (C) edge [bend right=10] node [right]{$4$} (A)
        
        (A) edge [bend right=10] node [above] {$1$} (E)
        (E) edge [bend right=10] node [below] {$2$} (A)
        (C) edge [bend right=10] node [above] {$1$} (F)
        (F) edge [bend right=10] node [below] {$2$} (C)
        
        (B) edge [bend right=10] node [below] {$1$} (G)
        (G) edge [bend right=10] node [above] {$2$} (B)
        (D) edge [bend right=10] node [below]  {$1$} (H)
        (H) edge [bend right=10] node [above] {$2$} (D);
\end{tikzpicture} \caption{$H_0,H_1,H_2,H_3:$ A family of parity games of high register-index }\label{fig-family}

\end{center}
\end{figure}

\subsection{The register-index as a measure of complexity}
Given the elusiveness of a polynomial algorithm for parity games, there is a rich line of research in algorithms that are polynomial for restricted classes of parity games. The size of the priority co-domain is perhaps the simplest way of measuring the complexity of parity games: many parity game algorithms are exponential in the number of priorities, and therefore polynomial on parity games with a fixed number of priorities.
There are also many graph-theoretic restrictions, such as bounded tree-width, bounded clique-width and bounded entanglement, that allow for polynomial algorithms. These two classes of restrictions are orthogonal in the sense that the size of the priority assignment is agnostic to the underlying graph and vice-versa.

In contrast, the register-index can be seen as a measure of complexity that takes into account both the complexity of the priority assignment and the structure of the underlying parity game graph. Since for bounded $k$ the register game over a game $\G$ is just a parity game of polynomial size in $|\G|$, solving parity games of bounded register-index is polynomial. Since we may not know the register-index in advance, by solving both $\rge{k}{\G}$ and $\rgo{k}{\G}$ up to fixed $k$, we obtain a parameterised polynomial algorithm which solves games of register-index up to $k$, and does not return an answer for other arenas.

From the analysis in the previous sections, this algorithm seems likely to be effective on many---perhaps most---reasonable parity games. In particular, parameter $1$ suffices to solve various parity games that are hard for existing algorithms, including other quasi-polynomial algorithms; furthermore, the register-index is independent of measures such as tree-width and entanglement. In this sense, this algorithm is complementary to existing ones. In any case, as we shall see in the next section, a logarithmic parameter suffices to solve \emph{all} finite parity games, making this algorithm quasi-polynomial.

\section{Register Games and Finite Parity Games}\label{sec:ParityGames}

So far we have defined both parity and register games on potentially infinite arenas.
In this section we consider the special case of finite arenas, which is of particular interest for verification and synthesis.
We show that for finite parity games, the register-index is logarithmically bounded in the number of disjoint cycles, and by extension in the size of the game.
We deduce a quasi-polynomial algorithm for solving finite parity games.

\subsection{Logarithmic bound on the register index of finite parity games}
This section presents the key technical result, from which the various results in the sequel follow: the register-index is logarithmic in the number of disjoint cycles in a finite parity game.

For a finite directed graph $G$, we call \emph{\dc} the maximal number of vertex-disjoint cycles in $G$; the \dc of a parity game is that of the underlying directed graph.

In order to build strategies from subgame strategies, we define the notion of a \emph{defensive} strategy. Indeed, an arbitrary winning strategy in a subgame may output a finite number of large odd priorities; this is problematic when this strategy is used in a subgame that a play can enter infinitely often. Defensive strategies avoid high odd priorities from the start.

\begin{defi}[Defensive register-index]
For a subgame $\G$ with priorities bounded by $p$, in which \Eve has a winning strategy, a winning strategy $\strat$ for \Eve in $\rge{k}{\G}$ is \emph{defensive} if, from positions $(\pos,\rt,0)$ where $r_k\geq p$ and $r_k$ is even, a play that agrees with $\strat$ never outputs~$2k+1$.
$\G$ has defensive register-index $k$ if $k$ is the minimal integer such that \Eve has a defensive winning strategy in $\rge{k}{\G}$ from all positions.
\end{defi}

\begin{thm}\label{lem-log}
The register-index $k$ of a finite parity game of \dc $z$ is at most $1+\log z$.
\end{thm}

\begin{proof}
From the definition of register-index,
it suffices to consider the single-player parity games $\G$ induced by any winning strategy for \Eve in her winning region.
Observe that in the register games on $\G$, all positions with multiple successors belong to \Adam and \Eve's strategy consists of just choosing a register after each move in the underlying game.

We show by induction on the number of positions $n$ in $\G$ that the \emph{defensive} register-index of $\G$ is bounded by $1+ \log z$; the theorem follows. The base case, $n=1$, is trivial.
For the inductive step where $\G$ has at least two positions, let $G'$ be the game induced by positions of $\G$ of priority up to $p-2$, where $p$ is the even supremum of priorities that appear in $\G$. Then, let $G_1,\ldots,G_j$ be the \emph{maximal strongly connected subgames} of $G'$. Let $k_1,\ldots,k_j$ be their respective defensive register-indices, and $m$ the maximal among these.  
If there are no such subgames, then $p$ is the maximal priority in all cycles; then Eve wins defensively in $\rge{1}{\G}$ by choosing register $1$ when $p$ occurs and register $0$ otherwise. Since no priority higher than $p$ occurs, $r_1$ always contains either $p$ or the initial value of $r_1$; this strategy is therefore winning and defensive.

\begin{description}
\item[Case of a unique $i$ for which $k_i=m$ and $m>0$]
We show that the defensive register-index of $\G$ is no more than $m$. Since the \dc of $\G_i$ is no larger than $z$, and from the inductive hypothesis $m\leq 1+\log z$, this suffices. 

Eve's strategy in the $m$-register game on $G$ is as follows:  within a subgame $G_j$ she uses the bottom  $k_j$ registers to simulate her defensive winning strategy in the $k_j$-register game on $G_j$; elsewhere, she chooses register $m$ whenever $p$ occurs and register $0$ otherwise.

Assume that this strategy is played from a register-configuration where the top register contains an even priority greater or equal to the largest priority in $\G$. First observe that after choosing register $m$ upon seeing $p$, $r_m$ is set to $p$. Furthermore, this strategy encounters $p$ between any two entrances into $\G_i$, and outside of $\G_i$ \Eve chooses register $0\neq m$ whenever $p$ does not occur; therefore plays that agree with this strategy always enter $G_i$ with either $p$ or the larger even initial $r_m$-value in register $m$. Thus, since the strategy within $\G_i$ is defensive, it never outputs $2m+1$ in $\G_i$. Similarly, when $p$ occurs $\max(r_m,p)$ is always even: $r_m$ is either its initial value, $p$, or a smaller priority from $\G_i$.  Plays that agree with this strategy either eventually stay within a subgame $G_j$ where \Eve follows a winning strategy, or change subgames infinitely often. In the latter case, $p$ is seen infinitely often, thus producing $2m$ as output infinitely often. This strategy is therefore both defensive and winning for \Eve.

Note that if this strategy is played from a register-configuration in which the top register is not an even priority greater or equal to the highest priority in $\G$, then a play might output $2k_m+1$ once, but the values output infinitely often will still be the same as above, so the strategy is still winning.

\item[Case of $i,j$ where $i\neq j$ and $k_i=k_j=m$] 
We show that the defensive register-index of $\G$ is at most $m+1$. This suffices, since by the induction hypothesis, each of $\G_i$ and $\G_j$ has \dc at least $2^{m-1}$; then $\G$ has \dc at least $2^{m}$ as $G_i$ and $G_j$ are disjoint. 
\Eve's strategy in the $m+1$-register game on $\G$ is as follows:  in a subgame $G_i$, she uses registers up to $k_i$ to simulate a winning strategy in the $k_i$-register game on $G_i$; elsewhere, she chooses register $m+1$ when $p$ is seen, and register $0$ otherwise. This strategy is winning as a play either eventually stays within a subgame where \Eve is following a winning strategy, or it enters some subgame infinitely often. In this case, it  must see $p$ infinitely often; then $2(m+1)$ is the highest value output infinitely often since after the first occurrence of $p$ register $m+1$ permanently contains $p$. Furthermore, if the initial content of register $m+1$ is an even priority larger than any occurring in $\G$, then every time \Eve chooses register $m+1$, the output is even; the strategy is therefore defensive.

\item[Case of $m=0$] 
Eve can win the $1$-register game on $G$, using a strategy as above: within a subgame $\G_i$, she uses her strategy in the $0$-register game, and chooses register $1$ whenever $p$ is seen. This strategy is winning since a play either remains in a subgame and follows a winning strategy, or sees $p$ infinitely often and therefore outputs $2$ infinitely often, but does not output $3$ infinitely often. If initially $r_1$ contains an even priority greater or equal to the maximal priority in $\G$, this strategy never outputs $3$ and is therefore defensive.
\qedhere
\end{description}
\end{proof}

\begin{rem}
The logarithmic bound was shown in~\cite{lehtinen2018modal} with respect to the size of a finite parity game directly; in~\cite{BL18} it was strengthened to a logarithmic bound with respect to the maximal number of disjoint strongly connected components. Here we opt for the number of vertex-disjoint cycles, which is equivalent but more intuitive.
\end{rem}

\subsection{A quasi-polynomial algorithm for parity games}

The quasi-polynomial solvability of parity games then follows: 
to solve a parity game $\G$ of size $n$, one can always solve $\rge{1+\log n}{\G}$ instead.

\begin{cor}[Also \cite{calude2017deciding, jurdzinski2017succinct}]
Parity games are solvable in quasi-polynomial time.
\end{cor}
\begin{proof}
Let $n$ be the number of positions in a parity game $\G$ with $d$ distinct priorities and $m$ edges. The game $\rge{k}{\G}$ has $\mathcal{O}(nd^{k+1})$ positions, priorities up to $2k+1$ and $\mathcal{O}((m+nk)d^k)$ edges. Since $\rge{k}{\G}$ is presented with its priorities on its edges, for the complexity analysis we take the size of $\rge{k}{\G}$ to be $O(knd^k)$, to account for an additional vertex for each of the $O(knd^k)$ edges of significant priority---that is, those in $E_r$ for some~$r\in[0..k]$.

From Theorem \ref{lem-log}, the register-index of a parity game of size $n$ is at most $1+\log n$. Therefore solving parity games $G$ reduces to solving $\rge{k}{\G}$ where $k=1+\log n$.
The $\rge{k}{\G}$ game can then be solved with an algorithm exponential in the number of priorities of one's choice, say the small progress measure algorithm~\cite{jurdzinski2000small}, to obtain a quasi-polynomial algorithm that runs in time $2^{\mathcal{O}((\log n)^3)}$.
\end{proof}

Without further optimisations, the complexity of this algorithm is not competitive with respect to existing quasi-polynomial algorithms; its space-complexity, due to building $\rge{1+\log n}{\G}$, is also quasi-polynomial. We refer the reader to Parys's~\cite{Parys20} and to Daviaud, Jurdzi\'nski and Thejaswini's~\cite{strahler} recent work that shows how to bring both the space and time complexity down to match other quasi-polynomial algorithms.
 Furthermore, as discussed in Section \ref{sec:examples}, many games have low register-index, including those that exhibit worst-case complexity for other algorithms; on such games this approach is expected to work well. Section \ref{sec:conclusions} discusses further ideas to adapt this approach for practical usage.

One of the principal appeals of this technique for analysing parity games comes undoubtedly from the logic and automata-theoretic perspective.
In the sequel, we turn our attention to infinite parity games, or at least parity games of unbounded size, which we use to build automata transformations based on register games.

\section{Register Games and Parity Automata}\label{sec:ParityAutomata}

Parity automata are closely related to parity games. On one hand, the language of parity games (represented as infinite trees) with a fixed number of priorities in which Eve has a winning strategy is recognised by an alternating parity tree automaton. On the other hand, the acceptance game for alternating parity automata---both on trees and on words---is an infinite parity game.

In this section we first establish that register games are, like parity games, \emph{automata-definable}: the winning regions of $k$-register games on arenas with $d$ priorities are recognised by an alternating parity tree automaton of size $O(d^{k})$.
We then use register games on \emph{infinite} and \emph{unbounded} arenas to define a quasi-polynomial translation from alternating parity word automata into alternating weak word automata.

We finish with a discussion of how these techniques generalise to restricted classes of tree automata, on which we observe a partial collapse of the parity index hierarchy.\\

We begin this section by fixing notations and definitions for $\omega$-word and tree automata with parity acceptance conditions.

\subsection{Automata on Words, Trees and Games}

\subsubsection*{Words, trees, games and graphs}
A \emph{word} over $\Sigma$ is a (possibly infinite) sequence $w=w_{0}\con w_{1}\cdots$ of letters in $\Sigma$. We write $\Suf(w)$ for the set of suffixes of $w$ (which includes $w$ itself).
We consider a \emph{tree} to be an infinite directed rooted tree in the graph-theoretic sense; in particular, the children of a node are not ordered. A \emph{$\Sigma$-tree} $t$ (resp. $\Sigma$-graph) is a tree (resp. graph) together with a mapping of each of its nodes to a letter in $\Sigma$. We write $\Subtrees(t)$ for the set of $\Sigma$-subtrees of $t$ (which includes $t$ itself). A $\Sigma$-graph with an initial vertex unfolds into an infinite $\Sigma$-tree; a regular $\Sigma$-tree can be finitely represented by a $\Sigma$-graph (i.e., Kripke structures).

We will use a \emph{game alphabet} $\Gabc^d=\{A_i,E_i|i\in [0..d]\}$ to represent parity games with priorities up to $d$. A $\Gabc^d$-tree, or graph, is interpreted as the infinite or finite parity game on the arena consisting of the tree, or graph respectively, where positions labelled $E_i$ belong to \Eve, positions labelled $A_i$ belong to \Adam, and positions labelled $E_i$ or $A_i$ have priority $i$.

\subsubsection*{Automata}

Several different definitions of alternating tree automata exist. One variable is the branching type of the input trees: fixed or arbitrary, with or without an order on the successors. In the context of parity games, it is most natural to operate in a setting with arbitrary branching without an order on the successor nodes---that is, an automaton can not distinguish between the left and right branch of a binary tree, unless this is encoded in the labelling. Another difference in automata definitions is how flexible the transition condition is with respect to $\epsilon$-transitions and the combination of path quantifiers ($\Diamond, \Box$) and boolean connectives ($\lor, \land$). Usually, all of these definitions give the same expressiveness (\cite[Proposition 1]{Wil99} and \cite[Remark 9.4]{Kir02}), except for the case of very restricted automata, in which they do not \cite{BS18}. In our definition of alternating automata, there are no epsilon transitions and path quantifiers are applied directly on states.

An \emph{alternating parity tree automaton} is a tuple $\AutTuple{\Sigma,Q, \iota, \delta, \Omega}$ where $\Sigma$ is a finite alphabet, $Q$ is a finite set of states, $\iota\in Q$ is an initial state, $\delta:Q\times\Sigma\to \BP(\{\Diamond,\Box\}\times Q)$ is a transition function, $\BP(X)$ is the set of positive boolean formulas over $X$, and $\Omega:Q\rightarrow I$ is a \emph{priority assignment} that assigns priorities from $I=[0..i]$ or $I=[1..i]$, for some $i\in\Nat$, to states. 
Intuitively, given a state and a letter, the transition function returns a positive boolean formula that defines which states the
automaton can transition to, and whether it considers the next state at a non-deterministically chosen child ($\Diamond$), or at all of the children ($\Box$). Positive boolean formulas over $\{\Diamond,\Box\}\times Q$ are called \emph{transition conditions}. The transition graph of $\A$ is the graph $(Q,E)$ where $(q,q')\in E$ if $q'$ appears in $\delta(q,\alpha)$ for some $\alpha\in\Sigma$.

\emph{B\"uchi} and \emph{co-B\"uchi} automata are  special cases of parity automata in which $I=\{1,2\}$ and $I=\{0,1\}$ respectively.
An automaton is \emph{weak} if every strongly connected component in the transition graph consists of states with either only odd priorities or only even priorities.
Observe that a weak automaton can be seen as either a B\"uchi or co-B\"uchi automaton by using priorities $1$ and $2$, or $1$ and $0$, respectively.

The \emph{size} of an automaton is the maximum of the alphabet size, the number of states, the number of subformulas in the transition function, and the acceptance condition's \emph{index}, that is $|I|$. Observe that in alternating automata, the difference between the size of an automaton and the number of states in it can stem from a transition function that has exponentially many subformulas.\\ 

How exactly nondeterminism in tree automata is defined also varies in the literature. In general, it only concerns the boolean connectives of the transition condition and not the path quantifiers. We consider an alternating automaton to be \emph{nondeterministic} (resp.\ \emph{universal}) if its transition conditions only use the $\lor$ (resp.\ $\land$) connective, in addition to the path quantifiers $\Diamond$ and $\Box$.

A \emph{word automaton} is simply a tree automaton that operates on unary trees, that is, $\omega$-words. For word automata, the path quantifiers $\Diamond$ and $\Box$ are equivalent as there is always exactly one successor. 

We will also mention \emph{deterministic} word automata, in which the transition condition has no boolean connectives, and \emph{safety} automata, in which the only rejecting state is a rejecting sink.

The \emph{class} of an automaton is determined by its transition mode (deterministic, nondeterministic, or alternating), its acceptance condition, and whether it runs on words or trees.
We often abbreviate automata classes by  acronyms in $\{$D, N, A$\} \times \{$W, B, C, P$\} \times \{$W, T$\}$. The first letter stands for the transition mode; the second for the acceptance-condition (weak, B\"uchi, co-B\"uchi, and parity); and the third indicates whether the automaton runs on Words or on Trees. For example, AWW stands for an alternating weak automaton on words.


We define the semantics of automata directly in terms of their acceptance games.

\begin{defi}[Acceptance game]\label{def:ModelCheckingGame}
Given a $\Sigma$-tree $t$ and an APT $\A=\AutTuple{\Sigma,Q,\iota,\delta,\Omega}$, the acceptance game $\game{t}{\A}$ is the following parity game:
\begin{itemize}
\item The set of positions is $\Subtrees(t) \times (Q \cup \BP(\{\Diamond, \Box\} \times Q)) $.
\item For $a\in \Sigma$, a $\Sigma$-tree $u$ whose root is labeled $a$, transition conditions $b$ and $b'$, and state $q\in Q$, there is an edge from:
\begin{itemize}
\item $(u,q)$ to $(u, \delta(q,a))$
\item $(u,b\vee b')$ to $(u,b)$ and $(u, b')$
\item $(u, b\wedge b')$ to $(u,b)$ and $(u, b')$
\item $(u,\Diamond q)$ to $(u',q)$, for every child $u'$ of $u$.
\item $(u,\Box q)$ to $(u',q)$, for every child $u'$ of $u$.
\end{itemize}
\item Positions $(\Box q,u)$ and $(b\wedge b',u)$ belong to Adam; other positions belong to Eve.
\item A position $(u,b)$ is of priority $\Omega(b)$ if $b$ is a state $q$, 
and 0 otherwise.
\end{itemize}
Observe that the  acceptance games of regular, i.e., finitely representable, trees are finite since these trees only have finitely many distinct subtrees. For a Kripke structure $S$, we write $\game{S}{\A}$ for the acceptance game $\game{t}{\A}$ where $t$ is the tree represented by $S$, and observe that it is finite.
\end{defi}

We say that an APT $\A$ with initial state $\iota$ accepts a tree $t$ if and only if Eve has a winning strategy in the acceptance game $\game{t}{\A}$ from $(t,\iota)$. A tree automaton accepts or rejects a graph according to whether it accepts or rejects its infinite tree unfolding. The set of trees accepted by $\A$ is called the language of $\A$, denoted by $L(\A)$. Two automata are equivalent if they recognise the same language.

As a word is a special case of a tree, Definition~\ref{def:ModelCheckingGame} and other results that will be shown on acceptance games apply also to word automata. In the word setting, the presentation of the acceptance game can be simplified: $\Diamond q$ and $\Box q$ are equivalent, and may thus be written~$q$, and $Q \cup \BP(Q)$ is equal to $\BP(Q)$.

\subsection{Automata for Register Games}
We define for every APT $\A$ and positive integer $k$, the parameterised version $\A_k$, which is an APT that will be shown to accept a tree $t$ if and only if Eve wins the $k$-register game on $\game{t}{\A}$ starting from $(t,\iota)$. The idea is to emulate the $k$-register game by keeping track of register configurations with a tuple $\rt \in I^{k+1}$ that is updated according to which priorities are seen and Eve's register choices, which are represented as nondeterministic choices in $\A_k$. The outputs are captured by the priorities of the states of $\A_k$. Here we note a slight subtlety: In the $k$-register game on $\game{t}{\A}$, Eve chooses registers not only at positions $(u,q)$ where $q$ is a state of $\A$, but also at positions $(u,b)$ where $b$ is a boolean formula. In $\A_k$ only states can have priorities (i.e., there can only be one priority per move in $t$) so we aggregate the outputs from these choices between two states by taking the largest output into the priority of the next state---this is the third element $p\in [0..2k+1]$ of the states of $\A_k$.

\begin{defi}
Given an APT $\A=\AutTuple{\Sigma,Q,\iota,\delta,\Omega}$ with $\Omega:Q\rightarrow I$ and a positive integer $k$, we define an APT $\A_k=\AutTuple{\Sigma,Q',\iota',\delta',\Omega'}$ as follows:
\begin{itemize}
\item $Q' = Q\times I^{k+1} \times [0..2k+1]$
\item $\iota' = (\iota,(0,..,0),0)$
\item $\Omega'$: For every $q\in Q, \rt \in I^{k+1}$, and $p\in[0..2k+1]$, we have $\Omega'(q,\rt,p)=p$.

\item $\delta'$: For every $q\in Q, \rt \in I^{k+1}, p\in[0..2k+1]$, and $a\in\Sigma$, we have $\delta'((q,\rt,p),a) = \bigvee_{i\in [0..k]}\move(\delta(q,a),\update_i(\rt,\Omega(q)),\max(r_i,\Omega(q)))$ where for transition conditions $b,b'$:
\begin{itemize}
\item
$\update_i(\rt,p)= \rt'$ where $r_j'=\max(r_j,p)$ for $j>i$, $r_i=p$ and $r_j=0$ for $j<i$;
\item $\move(\Diamond q',\rt,p)=\bigvee_{i\in [0.. k]} \Diamond (q', \update_i(\rt,0),m)$ where $m=2i$ if $\max(r_i,p)$ is even; $2i+1$ otherwise;
\item $\move(\Box q',\rt,p)=\bigvee_{i\in [0.. k]} \Box (q', \update_i(\rt,0),m)$
where $m=2i$ if $\max(r_i,p)$ is even; $2i+1$ otherwise;
\item $\move(b \wedge b',\rt,p)=\bigvee_{i\in [0..k]} \move(b, \update_i(\rt,0),\max(r_i,p))\wedge \move(b', \update_i(\rt,0),\max(r_i,p))$
\item $\move(b \vee b',\rt,p)=\bigvee_{i\in [0..k]} \move(b, \update_i(\rt,0),\max(r_i,p))\vee \move(b', \update_i(\rt,0),\max(r_i,p))$
\end{itemize}

\end{itemize}

\end{defi}

\begin{lem}\label{lem-Ak}
Given an APT $\A$ and a positive integer $k$, the parameterised APT $\A_k$ accepts a tree $t$ if and only if Eve wins the $k$-register game on $\game{t}{\A}$ from $(t,\iota)$.
\end{lem}
\begin{proof}
Recall that $\A_k$ accepts a tree $t$ if and only if Eve wins the parity game $\game{t}{\A_k}$ from the initial position, which consists of $\A_k$'s initial state and $t$'s root. Thus, we should show that Eve wins $\game{t}{\A_k}$ from the initial position if and only if she wins the $k$-register game on $\game{t}{\A}$ from the initial position. The intuition for the equivalence between the games is that $\game{t}{\A_k}$ encodes the $k$-register game on $\game{t}{\A}$ by having the register-configuration in the state space, Eve's register choices as new disjunctions, and the highest output between two states as priorities.

Positions of $\game{t}{\A_k}$ are in 
$$ \Subtrees(t)\times \Big((Q\times I^k \times [1..2k+1]) \cup \BP(\{\Diamond, \Box\} \times Q\times I^k \times [1..2k+1])\Big) $$
while positions of the $k$-register game on $\game{t}{\A}$  are in
$\Subtrees(t)\times \Big(Q \cup \BP(\{\Diamond, \Box\} \times Q)\Big) \times   I^k$.

$\game{t}{\A_k}$ begins at $(t,(\iota,(0,..,0),0))$ while the $k$-register game on $\game{t}{\A}$ begins at $(t,\iota,(0,..,0))$.

Now, observe that in the $k$-register game on $\game{t}{\A}$ at position $(u,q,\rt)$, for a state $q\in Q$ and a $\Sigma$-tree $u$ whose root is labeled $a$, Eve has to choose a register $i\in [0.. k]$ before the parity game proceeds to $(u,\delta(q,a))$ with register configuration $\rt'$, which is exactly $\update_i(\rt,\Omega(q))$. Similarly in $\game{t}{\A_k}$, from $(u,q,\bar x,p)$ Eve chooses $i\in [0..k]$ before proceeding with the transition condition $\delta(q,a)$ with register configuration $\update_i(\rt,\Omega(q))$. In the register game on $\game{t}{\A}$, this produces an output $2i$ if $\max(r_i,\Omega(q))$ is even; $2i+1$ otherwise. In $\game{t}{\A_k}$ the value $\max(r_i,\Omega(q))$ is kept in the state of the successor position.

When the games proceed, in both the register game on $\game{t}{\A}$ and in the parity game on $\game{t}{A_k}$, the decision falls to \Adam if $\delta(q,a)$ is a conjunction or prefixed by $\Box$ and to \Eve otherwise. Then, between each subformula of the transition condition, Eve again has a choice of register in both games. In the register game on $\game{t}{\A}$, this produces an output $2i$ if $\max(r_i,0)$ (that is $r_i$) is even; $2i+1$ otherwise. The largest output between two positions in $Q\times \Subtrees(t)$ of $\game{t}{\A}$ is recorded via the third element of the automata-component of $\game{t}{\A_k}$, which determines the priority of the successor state.

Then a winning strategy for Eve in one game translates into a winning strategy for Eve in the other game by mapping \Eve's choices of registers in the register game on $\game{t}{\A}$ to her choices in $\game{t}{\A_k}$ at disjunctions over $[0..k]$ in the transition condition of $\A$ and vice versa, and mapping her choices in the underlying parity game $\game{t}{\A}$ to the remaining choices in $\game{t}{\A_k}$, and vice-versa.
\end{proof}

\begin{rem}
The $\mu$ calculus-minded reader can get an idea of the equivalent modal $\mu$-calculus formula from~\cite{lehtinen2018modal} where the definability of register games was presented from a logic-perspective, rather than in terms of automata. The following proof in particular uses the automata-equivalent of the canonical modal $\mu$-calculus formulas that witness the strictness of the alternation hierarchy~\cite{bradfield98strict}.
\end{rem}


\begin{cor}
There is a parity automaton of size $d^{O(\log n)}$ that, over $\Gabc^d$-graphs of size $n$ and priorities up to $d$, recognises the parity games in which Eve has a winning strategy.
\end{cor}

\begin{proof}
Let $\P^d$ be the alternating parity automaton on $\Gabc^d$-trees that, over parity game graphs with priorities up to $d$, recognises parity games with priorities up to $d$ in which \Eve has a winning strategy. This is simply an automaton with states $0,1\dots d$ where the priority of a state $i$ is $i$ and the transition condition is $\delta(q,E_i)=\Diamond i$ and $\delta(q,A_i)=\Box i$. 
The acceptance game $\game{t}{\P^d}$ for any $\Gabc^d$-tree $t$ is identical to $t$ once we collapse positions with a unique successor. It follows that $t$ and $\game{t}{\P^d}$ have identical register-index. Then, from Theorem \ref{lem-log} and Lemma \ref{lem-Ak}, over $\Gabc^d$-graphs of size up to $n$ and priorities up to $d$, $\P^d_{\log n+1}$ recognises the same language as $\P^d$, which recognises the games in which Eve has a winning strategy.
\end{proof}

This corollary casts the quasi-polynomial complexity of parity games in terms of their \textit{descriptive complexity}---that is,  in terms of the complexity of the logical formalisms able to capture their winning regions. Note that here we only consider the complexity of recognising parity games \emph{with a bounded number of priorities} (here bounded by the size of the game) in which \Eve has a winning strategy; for parity games with an unbounded number of priorities, this task is beyond the expressivity of parity automata and the $\mu$ calculus~\cite{dawar2008descriptive}.

\begin{rem}
Observe that the register-index is bisimulation-invariant.
Indeed, 
 we can build an automaton that accepts exactly parity games with priorities up to $d$ that have register-index at most $k$.
Since alternating parity automata, like formulas of the modal $\mu$-calculus, are bisimulation invariant, so is the register-index.  
\end{rem}

\subsection{Register Games and Word Automata}\label{sec:words}

We turn our attention to automata over $\omega$-words. On words, deterministic parity automata, non-deterministic B\"uchi automata and alternating weak automata all are expressive enough to capture all $\omega$-regular languages~\cite{Lin88}. 
However, our understanding of the trade-offs in conciseness between different acceptance conditions is incomplete.  For instance, an ABW can be turned into an AWW with quadratic size blow-up~\cite{KV01c}, but the current corresponding lower bound is at $\Omega(n\log n)$. Until recently the best translation from APW into AWW was exponential~\cite{KV98b}; here we improve it to quasi-polynomial; the lower bound remains at $\Omega(n\log n)$.

In this section we use register games to define a quasi-polynomial translation from APW to AWW. It is based on observing that the acceptance parity games of an automaton on an ultimately periodic word have a register-index that depends solely on the automaton and is at most logarithmic in its size. Given that equivalence over ultimately periodic words implies equivalence over all words~\cite{mcnaughton1966testing}, this will suffice to turn APW  into equivalent APW with a logarithmic number of priorities and of quasi-polynomial size; from there, Kupferman and Vardi's classic translation into weak automata~\cite{KV98b} produces a weak automaton which is also of only quasi-polynomial size.\\

We begin by showing that the \dc of $\game{t}{\A}$, for a tree $t$ that represents the computations of a finite Kripke structure with \textit{minimum feedback vertex set} size at most $d$ in each of its (maximal) strongly connected components (SCC), is linear in $d|\A|$. Then, since the register index is logarithmic in the \dc, it is not the number of states in a Kripke structure that influences the register index of the acceptance game, but rather its minimum feedback vertex set. 
 For Kripke structures representing ultimately periodic words, this measure is $1$. While a weaker lemma based directly on the single cycle of a Kripke structure representing an ultimately periodic word $w$ would suffice, the next section will use this stronger lemma.

\begin{defi}
Given a directed graph $G$, a feedback vertex set is a set of vertices that contains at least one vertex of every cycle in $G$.
Let the \emph{fvs-size} of $G$ be the size of a minimum feedback vertex set of $G$.
\end{defi}


\begin{lem}\label{lem-fvs}
Given a Kripke structure $S$ with fvs-size up to $d$ in each of its SCCs and an APT $\A$ with $n$ states, the parity game $\game{S}{\A}$ has register index at most $1+\log dn$.
\end{lem}

\begin{proof}
We first note that the register index of a parity game is the maximal one among its SCCs.
We show that the dc-size of any SCC in $\game{S}{\A}$ is at most $dn$. Then, by Lemma~\ref{lem-log}, its register index is at most $1+\log dn$.

First note that every SCC of $\game{S}{\A}$ results from a single SCC of $S$ and a single SCC of $\A$. It is therefore enough to assume that both $S$ and $\A$ consist of a single SCC.

We consider the graph $H$ that is derived from $\game{S}{\A}$ by ignoring the intermediate positions $(v,b)$, where $b$ is a boolean formula between positions of the form $(v,q)$, for a state $q$. That is, let $H$ be the graph consisting of just the vertices $(v,q)$ of $\game{S}{\A}$, where $q$ is a state. The edges of $H$ connect positions $(v,q)$ and $(v',q')$ if $(v',q')$  is reachable from $(v,q)$ in $\game{S}{\A}$ directly, that is, with a path which does not visit yet another position $(v'',q'')$ where $q''$ is a state.

Observe that $\game{S}{A}$ has \dc no larger than that of the graph $H$: a set of disjoint cycles in $\game{S}{A}$ induces a set of disjoint cycles in $H$.

Let $F$ be a feedback vertex set of $S$, having up to $d$ vertices. Since every cycle $C$ in $H$ corresponds to some cycle of $S$, it must contain a vertex $(v,q)$, such that $q$ is a state of $\A$ and $v\in F$. Hence, for every $v\in F$, there are up to $n$ vertex-disjoint cycles in $H$ that have a vertex of the form $(v,q)$. Therefore, there are up to $|F|n \leq dn$ disjoint cycles in $H$, and therefore in $\game{S}{\A}$.
\end{proof}

In particular, given an ultimately periodic word and an APW $\A$ with $n$ states, since the fvs-size of a lasso Kripke structure representing a word is $1$, the parity game $\game{w}{\A}$ has register-index at most $1+\log n$.
Then $\A$ is equivalent to its $1+\log n$ parameterised version. On can then obtain an equivalent weak automaton of quasi-polynomial size by applying Kupferman and Vardi's transformation~\cite{KV98b} to $\A_{1+\log n}$ instead of $\A$.

\begin{lem}\label{lem-eq}
Every APW $\A$ is equivalent to its parameterised version $\A_k$, for $k=1+\log |\A|$.
\end{lem}
\begin{proof}
As two $\omega$-regular languages are equivalent if they agree on all ultimately periodic words \cite{mcnaughton1966testing}, it suffices to argue that $\A$ and $\A_k$ agree on ultimately periodic words.

From Lemma \ref{lem-Ak}, $\A_k$ accepts an ultimately periodic word $w$ if and only if if Eve wins the $k$-register game on $\game{S_w}{\A}$ where $S_w$ is a finite Kripke structure representing $w$. Since $S_w$ is a lasso, and it has fvs-size $1$. From Lemma~\ref{lem-fvs}, Eve wins the $k$-register game on $\game{S_w}{\A}$ exactly when Eve wins the parity game on $\game{S_w}{\A}$, that is, when $\A$ accepts~$w$.
\end{proof}

\begin{thm}\label{thm:qp}
There is a translation of alternating parity word automata into alternating weak word automata incurring at most a quasi-polynomial size increase. In particular, every APW $\A$ of size (resp.\ number of states) $n$ is equivalent to an AWW of size (resp.\ number of states) $2^{O((\log n)^3)}$.
\end{thm}
\begin{proof}
From Lemma \ref{lem-eq}, an APW $\A$ with $n$ states and $d$ priorities is equivalent to its parameterised APW $\A_k$ for $k=1+\log n$, having $n \cdot d^k \cdot (2k+1)$ states and $2k+1$ priorities. The automaton $\A_k$ can then be turned into a weak automaton using standard techniques~\cite{KV98b} with a $O(m^{d'})$ blow-up, where $m$ is the number of states and $d'$ the number of priorities, which yields an AWW with $2^{O((\log n)^3)}$ many states, since $m$ is here $O(kn^{k+1})$, which is in $2^{O((\log n)^2)}$ and $d'=2k+1 $ is $ O(\log n)$.

If the size of $\A$ is dominated by the size $e$ of its transition function, namely when $e>n$, observe that the parameter $k$, the number of states in $\A_k$, and the number of priorities in $\A_k$ do not depend on $e$, while the size of $\A_k$'s transition function is  $O(k^2ed^k)$ is in $2^{O((\log e)^2)}$. Since the translation in \cite{KV98b} does not blow up the transition-function size more than it blows up the number of states, we end up with an AWW of size in  $2^{O((\log e)^3)}$.
\end{proof}

\subsection{Register Games and Tree automata}
While the parity hierarchy collapses for alternating word automata, that is, AWWs suffice to recognise all word languages recognised by APWs, it is strict for alternating tree automata \cite{Lin88,bradfield98strict,BL18}: For every positive integer $n$, there is an APT $\A$ with $O(n)$ states and $O(n)$ priorities, such that there is no APT equivalent to $\A$ with less than $n$ priorities. Hence the weak condition does not suffice to capture all tree-languages captures by parity automata, so a general translation from APT to AWT does not exist. However, even for APT that are equivalent to an AWT, such a translation is at least exponential~\cite{BL18}, in contradistinction to the quasipolynomial translation of APW to AWW.



A natural question is then to consider the index hierarchy of alternating automata on entities that are ``between'' words and trees---that is, infinite trees generated by Kripke structures that are more complex than lassos, but more restrictive than trees. Examples of such classes are ``flat Kripke structures'' \cite{DDS12}, in which every SCC has a single cycle, ``weak Kripke structures'' \cite{KF11}, in which SCCs cannot have two vertex-disjoint cycles and Kripke structures with finite Cantor--Bendixson rank~\cite{blumensath2018bisimulation}.

We show that the parity hierarchy \emph{collapses logarithmically} for alternating automata on the classes of Kripke structures with bounded fvs-size in each of its SCCs. That is, for every APT $\A$ with $n$ states, there is an APT $A'$ with $O(\log n)$ priorities, such that $\A$ and $\A'$ are equivalent with respect to trees generated by Kripke structures of which the SCCs have bounded fvs-size. Kripke structures of which the SCCs have bounded fvs-size subsume Kripke structures that are flat, weak or of finite Cantor--Bendixson rank.

\begin{thm}\label{thm:Fvs}
Every APT $\A$ with $n$ states is equivalent to an APT $\A'$ with $1+\log dn$ priorities with respect to Kripke structures with fvs-size of up to $d$ in each of their SCCs.
\end{thm}
\begin{proof}
Let $k=1+\log dn$.
Setting $\A'$ to be the $k$-parameterised version $\A_k$ of $\A$ fits the bill. Indeed:
$\A$ accepts a Kripke structure $S$ iff Eve wins $\game{S}{\A}$. 
By Lemma~\ref{lem-fvs}, for every Kripke structure $S$ with fvs-size of up to $d$ in each of its SCCs, the register index of $\game{S}{\A}$ is at most $k$.
Hence, Eve wins $\game{S}{\A}$ iff she wins the $k$-register game on $\game{S}{\A}$, and by Lemma~\ref{lem-Ak}, this happens iff $A_k$ accepts $S$.
\end{proof}


\begin{cor}
Every APT with $n$ states is equivalent to an APT with $O(\log n)$ priorities with respect to Kripke structures whose SCCs are bounded in one of the following measures: minimum feedback vertex set, minimum feedback edge set, maximal number of vertex-disjoint cycles, and maximal number of edge-disjoint cycles.
\end{cor}
\begin{proof}
Observe that the minimum feedback vertex set is always smaller than or equal to the minimum feedback edge set, and likewise the maximal number of vertex-disjoint cycles is smaller than or equal to the maximal number of edge-disjoint cycles. Hence, it is enough to consider the vertex-version of these measures.

The result with respect to the minimum feedback vertex set follows directly from Theorem~\ref{thm:Fvs}. 
The result with respect to the maximal number of disjoint cycles follows from Theorem~\ref{thm:Fvs} and the Erd\H{o}s-P\'{o}sa Theorem \cite{EP65}, which states that the maximal number of disjoint cycles is logarithmically bounded in the minimum feedback vertex set.
\end{proof}

It may be the case that this collapse of the index-hierarchy on these structures goes further, all the way to weak, as on words, combining techniques from Kupferman and Vardi~\cite{KV98b} and Daviaud, Jurdzi\'nski and Lehtinen~\cite{DJL19} for example, or to priorities $\{1,2,3\}$ as for trees with a countable number of infinite branches~\cite{thin}. We leave this as an open problem, as well as the related question of what is the simplest class of trees on which the hierarchy is strict.

\section{Other Quasi-polynomial Automata}\label{sec:other}

In this section we discuss the relationship between our approach and the notion of \textit{separating automata} for solving parity games.
Czerwi\'nski et al.~\cite{czerwinski2019universal} have argued that the existing quasi-polynomial parity game algorithms can all be seen as separating automata. We consider the converse question: we make explicit how separating automata imply a tree automaton that recognises the winning regions of parity games of bounded size and propose a criterion delineating when a separating automaton also implies a translation of alternating parity automata into weak automata.



\subsection{From separating automata to tree automata}
Boja\'nczyk and Czerwi\'nski consider automata that separate plays that agree with winning strategies for each player in parity games of size up to $n$ with up to $d$ priorities~\cite{toolbox}.
Finding a deterministic safety automaton of size $f(n,d)$ that separates these word languages suffices to solve parity games in time polynomial in $f(n,d)$.
If the separation condition is strengthened to a separation between the language of plays that agree with a positional winning strategy for \Eve in some parity game of size $n$ with $d$ priorities and the language of plays that do not satisfy the parity condition, then a quasi-polynomial lower bound applies~\cite{czerwinski2019universal}.

From a deterministic separating word automaton,  one can build a tree automaton that recognises parity game arenas of size up to $n$ in which Eve has a winning strategy:

\begin{prop}\label{prop:gfg}
Let $\A$ be a DPW over the alphabet $[0..\maxp]$ that:
\begin{itemize}
\item Accepts words that agree with a positional winning strategy for \Eve in some parity game of size up to $n$ with up to $\maxp$ priorities, and
\item Rejects words that agree with a positional winning strategy for \Adam in some parity game of size up to $n$ with up to $\maxp$ priorities.
\end{itemize}

Then, there is an APT $\A'$ of the same size and acceptance condition as $\A$ that recognises parity game arenas of size up to $n$ with up to $p$ priorities in which Eve has a winning strategy.
\end{prop}

\begin{proof}
While $\A$ operates on an alphabet of priorities, $\A'$ operates on the richer game alphabet $\Gabc^\maxp$ that also encodes the ownership of nodes.
$\A'$ has the same state-space, initial state and priority assignment as $\A$. The only difference is the transition function of $\A'$ which is simply $\delta'(q,E_i)=\Diamond \delta(q,i)$ and  $\delta'(q,A_i)=\Box \delta(q,i)$. In other words,
 $\A'$ gives the decision of which branch to choose to the player who owns the current position, but otherwise operates exactly as $\A$. Whichever player has a winning strategy in a parity game $\G$  of size $n$ with up to $d$ priorities encoded as a $\Gabc$-graph can copy their positional winning strategy in the acceptance game $\game{\G}{\A'}$; since $\A$ separates the plays resulting from such strategies, $\A'$ will accepts if and only if \Eve has a winning strategy in $\G$.
\end{proof}

\begin{rem} 
Proposition \ref{prop:gfg} extends to \emph{good for games} word automata~\cite{HP06,BL19}. These are (not necessarily deterministic) automata $\A$ over an alphabet $\Sigma$ for which, given any $\Sigma$-labelled arena $\G$---that is, a graph of which the positions are partitioned between two players and labelled with $\Sigma$---\Eve wins the synchronised product of $\G$ and $\A$ if and only if she has a strategy $\sigma$ in $\G$ such that every path that agrees with $\sigma$ forms a word accepted by $\A$. 
Good-for-games parity automata, like deterministic automata, can be turned into APTs in the way outlined in Proposition \ref{prop:gfg}. 

However, the register automaton that recognises whether Eve wins the $k$-register game over an arena of size $n$ with up to $d$ priorities, although separating plays that agree with a winning positional strategy for Eve from plays that are winning for Adam, is neither deterministic nor good for games. First, observe that for $k>0$, it is indeed trivially a separating automaton for words with up to $d$ priorities since all words have register index $1$. It is not good for games (of register index larger than $k$). It is only ``good for small games" in the sense that when interpreted on games, it is only guaranteed to operates correctly on parity games of size up to $2^{k-1}$.

\end{rem}

\subsection{From tree automata to automata transformations}

As we have seen in Section \ref{sec:words}, register games are also suited for handling the  parity games of unbounded size that stem from word automata, thus allowing for a quasi-polynomial translation from alternating parity word automata to alternating weak automata. The same generalisation is not immediate for  Calude et al.'s and Jurdzi\'nski and Lazi\'c's algorithms; in particular, a safety-automata based approach, as the one proposed by Czerwi\'nski et al.~\cite{czerwinski2019universal} is unlikely to suffice since safety automata are not as expressive as parity automata. This raises the following question: when does an automaton that recognises the winning regions of parity games of bounded size imply a translation from alternating parity word automata into weak automata? 

We propose infinite directed acyclic graphs (dags) of bounded width as a key ingredient. The acceptance games for word automata take this shape, which make them an interesting stepping stone between words and trees. We shall see that tree automata with an acceptance condition X (e.g. B\"uchi or weak) that recognise infinite parity game dags of \textit{bounded width} (rather than of bounded size) in which Eve has a winning strategy can be used to turn alternating parity word automata into word automata with the same acceptance condition X. The translation is easy: it consists of the synchronised composition of the two automata.

The technical details of the synchronised composition are cumbersome, but the idea is straight-forward: the synchronised composition of $\B$ and $\A$ is an automaton that accepts a tree $\tree$ if and only if $\A$ accepts the acceptance game of $\B$ and $\tree$ when viewed as a $\Gabc$-tree.

\begin{defi}[Synchronised Composition]
Let $\A=(\Sigma^A,Q^A,\iota^A,\delta^A,\Omega^A)$ be an APT with priorities up to $d$. Let $\B=(\Gabc^d,Q^B,\iota^B,\delta^B,\Omega^B)$ be an APT over the game alphabet $\Gabc^d$.

The synchronised product $\product{\B}{\A}$ is defined as:

\begin{itemize}
\item State space $Q^A\times Q^B$;
\item Alphabet $\Sigma^A$;
\item $\Omega(q_A,q_B)=\Omega^B(q_B)$.
\item Transition relation: $\delta((q_A,q_B),\alpha)=f^\alpha(q_A,\delta^B(q_B,\mathrm{label}(q_A))$ where:

\begin{AutoMultiColItemize}

\item $f^\alpha(q_A,\Diamond p) = \Diamond f^\alpha(\delta^A(q_A,\alpha),p)$
\item $f^\alpha(q_A,\Box p) = \Box f^\alpha(\delta^A(q_A,\alpha),p)$

\item $f^\alpha(b\wedge b',\Diamond p)= f^\alpha(b,p) \vee f^\alpha(b',p)$
\item $f^\alpha(b\vee b',\Diamond p)= f^\alpha(b,p) \vee f^\alpha(b',p)$
\item $f^\alpha(\Diamond q,\Diamond p)= \Diamond (q,p)$
\item $f^\alpha(\Box q,\Diamond p)= \Diamond (q,p)$
\item $f^\alpha(b\wedge b',\Box p)= f^\alpha(b,p) \wedge f^\alpha(b',p)$
\item $f^\alpha(b\vee b',\Box p)= f^\alpha(b,p) \wedge f^\alpha(b',p)$
\item $f^\alpha(\Diamond q,\Box p)= \Box (q,p)$
\item $f^\alpha(\Box q,\Box p)= \Box (q,p)$
\item $f^\alpha(b,c\vee c')=f^\alpha(b,c)\vee f^\alpha(b,c')$
\item $f^\alpha(b,c\wedge c')=f^\alpha(b,c) \wedge f^\alpha(b,c')$

\item $f^\alpha(b,p)=f^\alpha(b,\delta^B(p,\lab(b)))$ where $p\in Q^B$ and
\item $\mathrm{label}(b\wedge b')=\mathrm{label}(\Box q)=A_0$
\item $\mathrm{label}(b \vee b')=\mathrm{label}(\Diamond q)=E_0$
\item $\mathrm{label}(q)=E_{\Omega^{\A}(q)}$
  \end{AutoMultiColItemize}
\end{itemize}

Observe that $\product{\A}{\B}$ has the acceptance condition of $\B$. Indeed, the product construction preserves the priorities and the non-reachability of states of $\B$: if $q_B'$ is not reachable from $q_B$, then $(q_A,q_B')$ is not reachable from any $(q_A',q_B')$. Then, if $\B$ is weak, then $\product{\A}{\B}$ is weak. 
\end{defi}

\begin{lem}\label{lem:prod-eq}

$\product{\A}{\B}$ accepts a $\Sigma^A$-tree $t$ if and only if $\B$ accepts the acceptance game $\game{t}{\A}$ seen as a $\Gabc^d$-tree for $d$ the maximal priority in $\A$. 
\end{lem}
\begin{proof}
We show that $\game{\game{t}{\A}}{\B}$ and $\game{t}{\product{\A}{\B}}$ have the same winner.
The transition relation is designed so that 
 $\game{\game{t}{A}}{B}$ is identical to $\game{t}{\product{\A}{\B}}$. More precisely, we identify positions $((t,q_A),q_B)$ and $((t,b_A),b_B)$  in $\game{\game{t}{A}}{B}$ with $(t,(q_A,q_b))$ and $(t,f^\alpha(b_A,b_B))$ in $\game{t}{\product{\A}{\B}}$, respectively, where $\alpha$ is the label of $t$. We now show that this mapping preserves successors, position ownership and priorities; the preservation of winner follows.
\begin{description}
\item[Preservation of successors] We observe:

\begin{itemize}
\item $((t,q_A), q_B)$ has successor $((t,q_A),\delta^B(q_B,E_{\Omega(q_A)}))$

\item $(t,(q_A,q_B))$ has successor $(t,f^\alpha(q_A,\delta^B(q_B,E_{\Omega(q_A)})))$

\item $((t, b\wedge b'), \Diamond p))$ has successors $((t,b),p)$ and $((t,b'),p)$
\item $(t,f^\alpha(b\wedge b',\Diamond p))$ has successors $(t,f^\alpha(b,p))$ and $((t,f^\alpha(b',p)))$

\item $((t, \Diamond q), \Diamond p))$ has successors $((t',q),p)$ for all children $t'$ of $t$.
\item $(t,f(\Diamond q, \Diamond p))$ has successors $(t',(q,p))$ for all children $t'$ of $t$.

\item Similarly for other combinations of modalities and boolean operators.

\end{itemize}

\item[Preservation of ownership] Eve owns position $((t,b_a),b_B)$ in $\game{\game{t}{A}}{B}$ whenever $b_A$ is a disjunction or a $\Diamond$-formula, and positions with a unique successor. Similarly, Eve owns $(t,f^\alpha(b_a,b_B))$ in $\game{t}{\product{\A}{\B}}$ whenever $b_B$ is a disjunction or $\Diamond$-formula, and positions with a unique successor.

\item[Preservation of priorities] The priority of both $((t,q_A),q_B)$ and $(t,(q_A,q_B))$ is  $\Omega^B(p)$ and the priority of other positions is $0$.
\end{description}

Then, since $\game{\game{t}{\A}}{\B}$ and $\game{t}{\product{\A}{\B}}$ must have the same winner, $\product{\A}{\B}$ accepts a tree $t$ if and only if $\B$ accepts $\game{\A}{\B}$.
\end{proof}

Then, if for a class $\class$ of trees, for all $\tree \in \class$, $\B$ accepts $\G(\tree,\A)$ if and only if $\A$ accepts $\tree$, then $\product{\A}{\B}$ is equivalent to $\A$ over $\class$. In other words, to turn an automaton $\A$ into a simpler automaton that is equivalent on a particular class of structures, it suffices to find an automaton $\B$ with a simple acceptance condition that accepts only the winning acceptance games of $\A$ over that class of structures. In particular, to turn a parity word automaton into a B\"uchi or weak automaton, it suffices to study B\"uchi or weak automata that recognise the winner in acceptance games over words.

\begin{defi}
An infinite dag is of bounded width $m$ if its vertices can be partitioned into sets $L_0, L_1, L_2, \dots$ no larger than $n$ 
such that every edge goes from some layer~$L_i$ to the next
layer~$L_{i+1}$.
\end{defi}
\begin{cor}
An alternating parity automaton $\B_{(m,\maxp)}$  that over regular parity game dags of bounded width $m$ with up to $\maxp$ priorities recognises those in which Eve has a winning strategy induces a translation from alternating parity word automata to alternating automata with the number of priorities of $\B_{(m,\maxp)}$ and with state and size blow-up linear in $|B_{(m,\maxp)}|$.
\end{cor}

\begin{proof}
The acceptance game of a ultimately periodic word $w$ and an APW $\A$ with up to $\maxp$ priorities is a regular parity game dag of width $|\A|$ with up to $\maxp$ priorities; hence $\B_{(|\A|,\maxp)}$ recognises the winning regions of the acceptance games of $\A$ over ultimately periodic words. The automaton $\product{\A}{\B_{(|\A|,\maxp)}}$ is, from Lemma \ref{lem:prod-eq}, equivalent to $\A$ on ultimately periodic words and therefore over all words. The synchronised product $\product{\A}{\B_{(|\A|,\maxp)}}$ is therefore equivalent to $\A$ and has the acceptance condition of $\B_{(|\A|,\maxp)}$. The blow-up is linear in $\B_{(|\A|,\maxp)}$
\end{proof}

The register index approach is an instantiation of this method: indeed, the key measure we use to bound the register index, \dc, is also bounded in finite structures that can be unfolded into a dag of bounded width.
From these observations, it seems that reasoning about parity games of arbitrary size but with a bounded number of disjoint cycles, or infinite parity games of bounded width, rather than only finite parity games, is a key distinction between algorithms for solving finite parity games and translations from one type of automata into another. Indeed, this strategy was recently used by Daviaud, Jurdzi\'nski and Lehtinen to provide an alternative quasi-polynomial APW to AWW translation, based on universal trees, with size increase in $n^{O(\log \frac{d}{\log n})}$~\cite{DJL19}.

\section{Conclusions} \label{sec:conclusions}

We have presented an automata-theoretic take on solving parity games in quasi-polynomial time based on the notion of register games. This perspective enabled us to go beyond finite parity games, and provide a quasi-polynomial translation of alternating parity automata into weak automata.

\subsubsection*{Solving parity games in practice} This article has focused on the automata-theoretic aspect of register games, rather than how they can be used to solve parity games in practice. We would therefore forgive the reader for questioning the practicality of this quasi-polynomial algorithm, which, among its less flattering features, has quasi-polynomial space complexity.

Recent work by Parys~\cite{Parys20}, as well as by Daviaud, Jurdzi\'nski and Thejaswini~\cite{strahler},  show that the register-game approach can be optimised to have  both time and space-complexity in line with the state-of-the-art.
Furthermore, there is much unexplored potential both for improving the practical viability of this algorithm, and for using the insights of this analysis to improve existing solvers.

Indeed, we have seen that building parity games with truly complex winning strategies---in the sense of requiring more than a small constant number of register in the register game---is subtle business: see the construction in Lemma \ref{example-high}. We therefore conjecture that the parameterised version of this algorithm, which solves parity games of register index $1$, is a plausible candidate for solving \emph{most} reasonable parity games in polynomial time, and perhaps also in practice. Increasing the parameter to $2$ or $3$ could ensure it only fails for purpose-built counter-examples.
The insight that even complex-looking parity games tend to have low register index could be useful in itself, for optimising existing solvers, which are not necessarily quasi-polynomial, but still more effective in practice.

Finally, the register index approach seems suited for a symbolic implementation: given a symbolically represented parity game, the $k$-register game can also be represented symbolically without significant blow-up. This would bypass the space-complexity of the algorithm for parity games that benefit from concise symbolic representations.

\subsubsection*{Further open problems} We have, throughout this article, pointed to some problems that remain open. The most obvious, of course, is the complexity of solving parity games. Another fundamental question is the conciseness gap between alternating parity and alternating weak word automata: a quasi-polynomial upper bound and a quasi-linear lower bound. One of the difficulties for closing this gap is the lack of techniques on one hand to prove lower bounds for alternating automata, and on the other to use alternations effectively to describe the winning regions of parity games.

We also mentioned the parity index problem, which is connected to the automata-theoretic concerns that this article touches upon. The question of where exactly, when moving from words to trees, the index hierarchy becomes strict, is particularly interesting.

\bibliographystyle{alpha}
\bibliography{bib}

\end{document}